\documentclass[12pt,final,
thmsa,
epsf,a4paper]
{article}
\usepackage{booktabs}

\newtheorem{theorem}{Theorem}

\newtheorem{corollary}{Corollary}
\newtheorem{definition}{Definition}

\newenvironment{proof}[1][Proof]{\emph{#1.} }{\  \hfill $\square $ \vspace{5 pt}}

\usepackage{natbib}
\usepackage{amssymb}
\usepackage[dvips]{graphics}
\usepackage{amsmath}

\usepackage{mathpazo}
\usepackage{enumerate}

\usepackage{amsfonts}

\usepackage{amssymb}

\usepackage{enumerate}

\usepackage{mathrsfs}

\usepackage[usenames]{color}
\usepackage{cancel}

\usepackage{pgf,tikz}

\usetikzlibrary{decorations.pathreplacing,angles,quotes}
\usetikzlibrary{arrows.meta}
\usetikzlibrary{arrows, decorations.markings}
\tikzset{myptr/.style={decoration={markings,mark=at position 1 with %
       {\arrow[scale=2,>=stealth]{>}}},postaction={decorate}}}

\usepackage{hyperref}
\hypersetup{
    colorlinks,
    citecolor=blue,
    filecolor=black,
    linkcolor=blue,
    urlcolor=blue
}

\usepackage{pgf,tikz}
\usetikzlibrary{arrows}

\usepackage[utf8]{inputenc}
\usepackage{amssymb, amsmath,geometry}
\usepackage{fancyhdr}

\usepackage{soul}

\usepackage{emptypage}

\newcommand*\samethanks[1][\value{footnote}]{\footnotemark[#1]}

\usepackage[T1]{fontenc}
\DeclareFontFamily{T1}{calligra}{}
\DeclareFontShape{T1}{calligra}{m}{n}{<->s*[1.44]callig15}{}
\DeclareMathAlphabet\mathcalligra   {T1}{calligra} {m} {n}

\begin{document}

\title{Dynamic matching games: stationary equilibria under varying commitments\thanks{Thanks \ldots. We acknowledge financial support
from UNSL through grants 032016 and 030320, from Consejo Nacional
de Investigaciones Cient\'{\i}ficas y T\'{e}cnicas (CONICET) through the grant
PIP 112-200801-00655, and from Agencia Nacional de Promoción Cient\'ifica y Tecnológica through grant PICT 2017-2355.}}


\author{Nadia Guiñazú\thanks{Instituto de Matem\'{a}tica Aplicada San Luis, Universidad Nacional de San
Luis and CONICET, San Luis, Argentina. Emails: \texttt{ncguinazu@unsl.edu.ar} (N. Guiñazú). \texttt{paneme@unsl.edu.ar} (P. Neme),  and \texttt{joviedo12@gmail.com} (J. Oviedo).}\and Pablo Neme\samethanks[2]  \and Jorge Oviedo} 

\date{\today}
\maketitle

\begin{abstract}

This paper examines equilibria in dynamic two-sided matching games, extending Gale and Shapley’s foundational model to a non-cooperative, decentralized, and dynamic framework. We focus on markets where agents have utility functions and commitments vary. Specifically, we analyze a dynamic matching game in which firms make offers to workers in each period, considering three types of commitment: (i) no commitment from either side, (ii) firms' commitment, and (iii) workers' commitment. Our results demonstrate that stable matchings can be supported as stationary equilibria under different commitment scenarios, depending on the strategies adopted by firms and workers. Furthermore, we identify key conditions, such as discount factors, that influence agents’ decisions to switch partners, thereby shaping equilibrium outcomes.

\bigskip

\noindent \emph{JEL classification:} C78, D47.\bigskip

\noindent \emph{Keywords:} Dynamic matching game, stationary equilibria, commitment, stable matchings 

\end{abstract}
\section{Introduction}

This paper extends the classical static matching model of \cite{gale1962college} to a dynamic, non-cooperative environment where firms and workers interact repeatedly. Our aim is to analyze how the structure of commitment between agents and their patience---both on the firm side and the worker side---affects the stability of the resulting matchings in decentralized markets. We consider a dynamic matching game where firms make offers to workers in each period, and workers individually decide whether to accept these offers or reject them. We present different scenarios of commitment and examine how these influence equilibrium strategies and the stability of the matching achieved over the long term. In particular, we explore the concept of \emph{subgame-perfect Nash equilibrium in stationary strategies}, which we refer to as \emph{stationary equilibrium}, within this dynamic framework.

Since the seminal work of \cite{gale1962college}, extensive literature has emerged on two-sided matching, primarily focusing on the notion of stability. A matching is deemed stable if all agents have acceptable partners and there exists no pair of agents---one from each side of the market---who would prefer to match with each other rather than remain with their current partners. Two-sided matching problems can be studied within both centralized and decentralized markets. In a centralized market, decisions and matchings are determined by a mechanism, such as the deferred acceptance algorithm, which is employed in systems like the National Resident Matching Program (NRMP). In this context, participants submit their preferences to a centralized entity that computes a stable matching based on specific rules.
In contrast, decentralized markets involve direct interaction among participants, who negotiate and reach agreements without a centralized intermediary. This raises critical questions regarding the ability of such markets to generate stable matchings. Our paper addresses these issues by studying equilibria in a dynamic and non-cooperative environment where firms and workers interact repeatedly. In each period, firms make offers to workers, who must individually decide which offer to accept.

A crucial aspect of dynamic models is the degree of commitment that agents can maintain with their partners over time. For instance, consider two possible scenarios: one where schools hire teachers, and another where sports teams hire athletes. Workers (teachers) may hold tenured positions in schools, implying that schools hold the commitment. Conversely, workers (athletes) may sign binding contracts with sports teams, implying that the commitment lies with the workers. Consequently, agents' actions may be influenced by the type of commitment assumed in the dynamic game.

Before delineating the game, it is essential to define what it means for an agent to be \emph{active} or \emph{inactive} in a given period. A worker is considered active if she does not hold a commitment with the agent she was matched with within the previous period or if she was unmatched in the previous period. Similarly, we define when a firm is active or inactive. 

The dynamic game analyzed in this paper is structured as a two-stage game for each period. In the first stage, firms simultaneously make offers to at most one worker. An active firm can extend an offer to any worker, while an inactive firm must retain its current employee. Firms remain unaware of the offers made by other firms in that period, but they are cognizant of the matching history from previous periods. Inactive firms are permitted to make offers to their current employees as well. 

In the second stage, each worker privately observes the offers received in the first stage, including the renewal offer from their current employer, if applicable. Workers do not know the offers made to others, but they are aware of the matchings from prior periods. Each worker can accept at most one offer and may choose to reject all offers. An active worker can accept any offer, while an inactive worker must accept the renewal offer from their current employer. 

Given the dynamic nature of the game and the assumption that agents do not hold beliefs about the actions of others, we employ the concept of stationary equilibrium as a solution concept. Recall that a strategy profile (one strategy for each agent) constitutes a stationary equilibrium if no possible deviation by an agent yields higher utility when considering the subgame involving only the active agents.

In dynamic games, the \emph{commitment} of agents plays a critical role. In this paper, we consider three scenarios: (i) when neither side of the market holds commitment, (ii) when only firms hold commitment, and (iii) when only workers hold commitment.

When both sides of the market possess commitment, the dynamic game reduces to a two-stage game where firms make offers, and workers either accept or reject them, concluding the game. This is due to the type of equilibrium we are studying, as in this context we do not entertain beliefs about the potential actions of other agents. Several studies have analyzed this setting, demonstrating that stable matchings can be supported by equilibria \citep[see][among others]{alcalde1997hiring,alcalde2000simple}.

In the case where only one side of the market holds commitment, we present two equilibria derived from two different strategies of firms—one more restrictive than the other. In the scenario where firms hold commitment (i.e., they offer tenured positions to workers), given a stable matching, if firms make offers to workers assigned under the stable matching and each worker accepts the offer yielding the highest utility, this strategy constitutes a stationary equilibrium that supports the stable matching.

Now, consider the first stage where firms make offers to the worker matched under the stable matching, and each worker accepts that offer. In this scenario, we allow firms to make new offers if they are active. Given a stable matching, if a worker seeks to improve her situation, she must resign from her current position with the firm she is assigned to under the stable matching and wait for a new offer. Due to the nature of firms' commitment, this worker must wait for a new offer, which incurs a cost referred to as the \emph{discount factor}. 
\cite{bonifacio2024counting} presents a re-stabilization process wherein a new and better offer always arrives for the worker who opts to resign. Additionally, they calculate the time required for the worker to receive this new offer. Once a worker decides to improve her situation and resigns from a firm, that firm becomes active. Since each active firm aims to maximize its utility, workers will only resign if the new offer yields higher utility, considering their discount factors; that is, their final utility is positive despite the discount factor. Thus, we can establish a threshold for discount factors such that if these factors exceed this threshold, it will not be worthwhile for workers to resign and await new offers. In this case, the proposed strategy constitutes a stationary equilibrium of the dynamic game.

In the scenario where workers possess commitment (i.e., workers sign binding contracts with firms), we demonstrate that, given a stable matching, two stationary equilibria support it, depending on the patience of firms in making offers. In this case, we also analyze two strategies for workers, one more restrictive than the other.
In the most restrictive case, similar to the scenario where firms hold commitment, we assume that the firms' strategy is always to offer to the worker employed under the stable matching. Workers accept the offer that provides them with the highest utility.
In the case of the less restrictive strategy, firms initially offer to the worker assigned under the stable matching. In subsequent stages, if firms observe that the stable matching was maintained in the previous period, they renew their offers to the same workers. Conversely, if firms observe that the stable matching was not formed in the prior period, they make offers to the workers employed under the firm-optimal stable matching. Both firm strategies, along with the workers' strategy of always accepting the offer that yields the highest utility, result in two stationary equilibria that support the same stable matching.


The closest paper to ours is that of \cite{diamantoudi2015decentralized}, who study a model in which firms make periodic offers to workers, who then decide whether to accept them. They also consider scenarios where agents hold commitment while entertaining beliefs about the behavior of other agents. They demonstrate that sequential equilibria in stationary strategies can lead to stable matchings. However, in the presence of commitment, some equilibria may yield unstable matchings, and the game does not always resolve in the first period, with certain matchings potentially materializing later. By imposing a consistency restriction on strategies, they prove that consistent outcomes are always stable matchings.

Similarly, \cite{haeringer2011decentralized} examine a dynamic game where payments depend solely on the final matching. In contrast to their approach, this paper considers a model in which matchings are formed in each period, and agents accumulate payoffs as the game progresses, highlighting the significance of both the duration of matchings and the final matching.

Other key works include those by \cite{konishi2008decentralized} and \cite{niederle2007matching}, who also study dynamic models by introducing factors such as wages and the duration of offers. However, these studies differ in that agents only observe the final matching, whereas our approach considers a continuous process of matching formation over time.

\cite{blum1997vacancy} develop an algorithm that identifies stable matchings when some agents are already paired. This provides an interesting perspective on how prior matchings can affect the dynamics of offers in repeated markets.

Finally, \cite{doval2022mechanism} introduces a revelation principle for single-agent dynamic mechanism design in Markov environments, where the agent's private information evolves stochastically and the designer commits only to short-term mechanisms. It demonstrates that all equilibrium payoffs can be achieved using flow-direct Blackwell mechanisms, allowing agents to truthfully report their types while aligning beliefs with the designer's equilibrium. This result simplifies the search for optimal mechanisms in dynamic settings, including dynamic Mirrlees and social insurance models.

The paper is structured as follows: In Section \ref{seccion de preliminares}, we present the preliminaries, the two-stage game, and the solution concept. Section \ref{seccion resultados} analyzes various commitment scenarios and their corresponding stationary equilibria. Finally, Section \ref{section concludings} offers concluding remarks.

\section{Preliminaries}\label{seccion de preliminares}

We consider two disjoint finite sets of agents, the set of firms $F$ and the set of workers $W$. An agent refers to either a firm or a worker. Each agent $i$ has a utility function $u_i$ such that
$$u_f:W \cup \{f\} \to \mathbb{R} \text{ for each } f \in F ,$$
$$u_w:F \cup \{w\}\to \mathbb{R} \text{ for each } w \in W.$$
We denote by $u_i(j)$  to agent $i$'s utility of being matched to agent $j$. Here, $u_i(i)$ is the agent’s utility of being unmatched, and we ask for utilities to fulfill that $u_i(i) =0$ for each $i$. If $u_i(j) \geq 0$ we say $j$ is acceptable to $i$. We assume ``strict'' utilities: $u_i(j) =u_i(k)$ only if $j =k$.

A matching is then a mapping $\mu: F\cup W \to F\cup W $ such that: (i) $\mu(f) \in W\cup \{f\}$ for each $f\in F$,  (ii) $\mu(w) \in F\cup \{w\}$ for each $w\in W$, and (iii) $\mu(i)= j$ if and only if $\mu(j)=i$ for each $i, j \in F\cup W $. Here, $\mu(i)$ denotes the agent with whom $i$ is matched and $\mu(i) =i$ means that agent $i$ is unmatched. Let $\mu_\emptyset$ denote the matching in which no one is matched.
A matching is individually rational if it asings to each agent $i$ a non-negative utility ($u_i(\mu(i))\geq 0$). We say that a matching  $\mu$ is blocked by a pair $(f, w) \in F\times W$ if
$$u_f (w) > u_ f (\mu( f )), \text{ and}$$
$$u_w( f ) > u_w(\mu(w)).$$
This means that $f$ and $w$  prefer each other to their partners under $\mu$. We say that a matching $\mu$ is stable if it is individually rational and has no blocking pair. The stable matching set is always non-empty. Moreover, there is always a stable matching that all firms agree to be the one that provides the highest utility, called the firm-optimal stable matching and denoted by $\mu_F$. Conversely, there is always a stable matching that all workers agree to be the one that provides the highest utility, called the worker-optimal stable matching, and denoted by $\mu_W$, \citep[see][]{gale1962college}.

In this paper, we consider a ``decentralized'' dynamic matching game. To formally define a dynamic matching game, we first define the payoff in each period. 

We consider discrete periods: $t=1, 2, 3,\ldots$. In each period, agents derive a payoff from the realized matching. The period-payoff function for agent $i$ is the utility function $u_i$ which is invariant over the periods. Each agent $i$ maximizes the discounted sum of period-payoffs,
$$\hat{u}_i=\sum_{t=1}^{\infty}\delta_i^{t-1}u_i(\mu^t(i)),$$
where $\delta_i \in (0, 1)$ is the discount factor and $\mu^t$ is the realized matching at period $t$. Now, we describe how $\mu^t$ is defined. In order to do so, note that the matching $\mu^{t-1}$ observed by each agent will determine, at the beginning of each period $t$, the set of firms and workers who cannot move in period $t$ depending on the type of commitment considered. We assume $\mu^0=\mu_\emptyset$, i.e., no one is matched before the initial period.

We denote by $F_c(\mu^{t-1})\subseteq F$ to the set of inactive firms at period $t$. These firms have committed themselves to their employees in $\mu^{t-1}$. This means that at period $t$, the firms in $F_c(\mu^{t-1})$ can neither fire their employees nor hire new ones. That is, their current employees have tenure and their jobs are protected. Conversely, we refer to the firms in $F \setminus F_c(\mu^{t-1})$ as active firms. These firms do not hold a commitment to their current employees, so they can fire their employees if they have any. Similarly, we denote by  $W_c(\mu^{t-1})\subseteq W$ the set of inactive workers at period $t$. This means that at period $t$, the workers in $W_c(\mu^{t-1})$ cannot switch their employers. Thus, we refer to the workers in  $W \setminus W_c(\mu^{t-1})$ as active workers. These workers do not hold a commitment to their employers and can leave their positions if employed.

Note that, depending on the characteristics of the sets $F_c$ and $W_c$ at a period $t$, there are three different scenarios to consider: 
\begin{description}
    \item[Case 1:] Two-sided commitment. All matched agents are inactive:
$$F_c(\mu^{t-1}) = \{f \in F : \mu^{t-1}( f ) \neq f \},\text{ and} $$
$$W_c(\mu^{t-1}) = \{w \in W : \mu^{t-1}(w) \neq w\}.$$
This means that once a firm and a worker are matched, they stay so permanently.
    \item[Case 2:] No commitment. All firms and workers are active regardless of the previous matching: $F_c(\mu^{t-1}) =W_c(\mu^{t-1}) =\emptyset.$ This means that firms can fire their employees, and workers can leave their current employers.

    \item[Case 3:] One-sided commitment. In this case, despite the market being symmetric, since only firms are making offers, we must analyze two subcases: when only firms hold commitment, and when only workers hold commitment. If only firms hold commitment, all matched firms are inactive, while all workers remain active:
$$F_c(\mu^{t-1}) = \{f \in F : \mu^{t-1}( f )\neq f \}, \text{ and}$$
$$W_c(\mu^{t-1})=\emptyset.$$
Workers cannot be fired but may switch to other firms when receiving new offers.
When only workers hold commitment, all matched workers are inactive, while all firms remain active:
$$F_c(\mu^{t-1}) =\emptyset, \text{ and}$$
$$W_c(\mu^{t-1})= \{w \in W : \mu^{t-1}( w )\neq w \}.$$
Workers must honor their employment contracts, even when receiving better offers. However, firms are allowed to lay off their workers in order to hire new ones, eventually.

\end{description}
 
Now, we are in a position to describe each period of the dynamic matching game. Each period will be decomposed into a two-stage game.
\begin{description}
     \item[First stage: ] Each firm simultaneously makes an offer to at most one worker. An active firm can make an offer to any worker while an inactive firm has no option but to keep its employee under $\mu^{t-1}$. Firms do not observe any offer made by another firm in the current period, but each firm observes the matching realized in previous periods. We consider that inactive firms make new offers to their current employees. The action of firm $f$, denoted by $o_f$, must fulfill that (i) $o_f \in W \cup \{ f \}$ if $f \in F \setminus F_c(\mu^{t-1})$ and (ii) $o_f = \mu^{t-1}( f )$ if $f \in  F_c(\mu^{t-1})$,
where $o_f=f$ means that $f$ makes no offer to any worker. Furthermore, we denote by  $O_f(\mu^{t-1})$ to the set of all possible offers that firm $f$ can make in period $t$ given $\mu^{t-1}$.
 \item[Second stage: ] Each worker $w$ privately observes the offers made to her in the first stage, denoted $O_w=\{f\in F: o_f=w \}$. Recall that $O_w$ includes the renewal offer from the current employer if $w$ has tenure. Workers do not observe any offer made to other workers in the current period, but each worker observes the entire matching realized in previous periods.  Thus, each worker simultaneously accepts at most one offer. An active worker $w$ can accept any offer or reject all offers. Inactive workers have no choice but to accept the renewal offers from their current employers. Thus, worker $w$'s response, denoted by $r_w$, must fulfill that (i) $r_w\in O_w \cup \{w\}$ if $w\in W\setminus W_c(\mu^{t-1})$ and (ii) $r_w = \mu^{t-1}(w)$ if $w\in W_c(\mu^{t-1}).$ Furthermore, we denote by $R_w(\mu^{t-1}, O_w)$ to the set of admissible responses for $w$.
\end{description}
Then, given the actions of firms and workers, the matching in period $t$, denoted by $\mu^t$, is defined by
$\mu^t (w) =r_w$ for each $w\in W$ and $\mu^t (f) =w$ if $f=r_w$. 
Thus, by definition, $\mu^t$ is individually rational at each period $t$. 
Therefore, a dynamic matching game is given by $(F,W,(u_i,\delta_i)_{i\in F\cup W}, F_c, W_c).$

Another two important concepts that we need to present used in our results are the histories and the strategies of the agents.
A history at the beginning of period $t$ is an ordered list of past actions, given by $$h_t =\left((o^{\tau}_f)_{f\in F},(r^\tau_w)_{w\in W}\right)_{\tau=1}^{t-1},$$
where $o^\tau_f$ is the offer made by firm $f$ in period $\tau=1, \ldots,  t-1$ and $r^\tau_w$ is the reply of worker $w$ in period $\tau$. After the first stage of period $t$, a history is given by $(h^t, (o^t_ f)_{f\in F})$, where $h^t$ is a history at the beginning of this period and $(o^t_f)_{f\in F}$ is the profile of offers made in this period.

The profile of replies $(r^\tau_w)_{w\in W}$ in $h^t$ contains the same information as the realized matching $\mu^t$, which becomes public information. Since offers are private information, players do not have complete information about the history. Each player observes only his private history, defined as follows.
A private history for firm $f$ in period $t$ is an ordered list $h^t_f= (\mu^0 = \mu_\emptyset, o^1_f, \mu^1, o^2_ f,\mu^2,\ldots, o^{t-1}_f ,\mu^{t-1})$, where $\mu^\tau$ is the matching realized in period $\tau$. While $\mu^\tau$ is public information, $o^\tau_f$ is private information. Let $H^t_f$ denote the set of private histories for $f$ in period $t$. Let $$H_f= \bigcup_{t=1}^{\infty}H^t_f$$ denote the set of all private histories for $f$.

A (pure) strategy of firm $f$ is a function $\sigma_f: H_f\to W\cup\{f\}$ such that for each $h^t_f \in H_f, \sigma_f(h^t_f)\in O_f(\mu^{t-1})$, where $\mu^{t-1}$ is the last entry of $h^t_f$.
Similarly, a private history for worker $w$ in the middle of period $t$ (when she makes a decision) is an ordered list
$h^t_w= (\mu^0=\mu_\emptyset, O^1_w, r^1_w,\mu^1, O^2_w, r^2_w,\mu^2,\ldots, O^{t-1}_w , r^{t-1}_w ,\mu^{t-1}, O^t_w)$,
where $O^\tau_w$ is the set of offers made to $w$ in period $\tau$ (including a renewal offer if any) and $r^\tau_w$ is her reply in that period. Let $H^t_w$ denote the set of all private histories for $w$ in period $t$. Let $$H_w=\bigcup_{t=1}^{\infty}H^t_w$$ denote the set of all private histories for $w$. A strategy of worker $w$ is then a function $\sigma_w: H_w\to F\cup\{w\}$ such that, for each $h^t_w\in H_w,$ $$\sigma_w(h^t_w)\in R_w(\mu^{t-1}, O^t_w),$$ where $\mu^{t-1}$ and $O^t_w$ are the last two entries of $h^t_w$.
A strategy profile $\sigma=(\sigma_i)_{i\in F\cup W}$ determines the payoff for each agent in the dynamic game. Given a dynamic matching game $(F,W,(u_i,\delta_i)_{i\in F\cup W}, F_c, W_c)$ and a strategy profile $\sigma$, $\mu_\sigma$ is the matching resulting of playing the dynamic matching game when agents declare the strategy profile $\sigma$, and  $\hat{u}_i(\mu_{\sigma})$ denotes the utility of each agent $i$. Moreover,  the utility for agent $i$ can be computed as $$\hat{u}_i(\mu_{\sigma})=\sum_{t=1}^{\infty} \delta^{t-1}_i u_i(\mu^t(i)).$$
 

In this paper, we assume that agents have no beliefs about the actions of other agents. Therefore, we consider \emph{subgame-perfect Nash equilibria in stationary strategies}, where each agent’s strategy depends only on the payoff-relevant state of the game.

\begin{definition}
A strategy profile $\sigma^\star$ is a Nash equilibrium if $u_i(\mu_{\sigma^\star})\geq u_i(\mu_{\sigma_i,\sigma^\star_{-i}})$  for each agent $i$ and each strategy $\sigma_i$ .\footnote{Here, $\mu_{\sigma_i,\sigma^\star_{-i}}$ denotes the matching resulting of playing the dynamic matching game when agent $i$ declare strategy $\sigma_i$ and the rest of agent $j$ declare strategy $\sigma^{\star}_j$.}
\end{definition}

Given a dynamic matching game,  a subgame that starts in time $\widetilde{t}$ is given by $$(F,W,(u_i,\delta_i)_{i\in F\cup W}, F_c, W_c)_{t=\widetilde{t}}^\infty.$$ Note that there are as many subgames starting at time $\widetilde{t}$ as there are histories at time $\widetilde{t}-1.$ A subgame is a part of the (original) dynamic matching game that starts at a point where the game's history up to that point is public knowledge among all agents and includes all subsequent decisions in the original game from that point onward.

\begin{definition}
    A strategy $\sigma^\star$ is a subgame-perfect Nash equilibrium if for each $\widetilde{t}$, $\sigma^\star$ is a Nash equilibrium for the subgame $(F,W,(u_i,\delta_i)_{i\in F\cup W}, F_c, W_c)_{t=\widetilde{t}}^\infty.$
\end{definition}

Since we do not consider the past histories of the game in this paper, we will employ stationary strategies to manage the dynamic problem more effectively. A crucial element in defining stationary strategies is the concept of continuation-equivalent games.

For an agent to decide which strategy to play at a given time, she will only consider the matching generated in the previous period. Therefore, if two different strategies result in equivalent matchings, the same strategy will be employed in the subsequent period in both cases.  Different matchings can induce distinct continuation-equivalent games depending on what commitment types agents possess. Given $\mu,\mu'\in \mathcal{M},$ the equivalence relation $\sim$ depends on the commitment structure of the game, and is defined as follows: 
\begin{itemize}
\item In the case of bilateral commitment, two matchings are continuation-equivalent if and only if the set of unmatched agents is identical. Formally, $\mu \sim \mu'$  if and only if $ \{ j \in F \cup W : \mu(j) = j \} = \{ j \in F \cup W : \mu'(j) = j \}$.  Agents who have been matched cannot change their partner for the rest of the game.
    \item In the no-commitment case, all matchings are continuation-equivalent. Formally, $\mu \sim \mu'$ for any two matchings $\mu$ and $\mu'$. In terms of commitment, the continuation of the game is the same regardless of what happened in the previous period.

    \item In the case of unilateral commitment, no two different matchings are continuation-equivalent. Formally, $\mu \sim \mu'$ if and only if $\mu= \mu'$. Even if the set of matched agents is the same, the continuation of the game depends on how the agents are currently matched.
\end{itemize}

Given a preference relation, we can define stationary strategies as follows. The strategy $\sigma_f$ of firm $f$ is stationary if for any pair of private histories $h_f = (\ldots, \mu)$ and $h'_f = (\ldots, \mu')$ (possibly with different lengths) we have that $\mu \sim \mu'$, then $\sigma_f(h_f) = \sigma_f(h'_f)$.
For workers' strategies, there is an additional requirement that the set of offers received in the current period is also identical. That is, the strategy $\sigma_w$ of worker $w$ is stationary if for any pair of private histories $h_w = (\ldots, \mu, O_w)$ and $h'_w = (\ldots, \mu', O'_w)$, if $\mu \sim \mu'$ and $O_w = O'_w$, then $\sigma_w(h_w) = \sigma_w(h'_w)$.

\begin{definition}
    A stationary equilibrium is a subgame-perfect Nash equilibrium in which all strategies are stationary.
\end{definition}

\section{Stationary equilibria in dynamic games}\label{seccion resultados}
In this section, we present the main results of this paper. This section consists of three subsections. In Subsection \ref{subseccion sin compromiso}, we consider a game where neither firms nor workers hold commitment. In this case, we show that any stable matching is the result of a stationary equilibrium. In Subsection \ref{subseccion firmas tienen compromiso}, we consider a market where only firms hold commitment. In this case, we show that, given a stable matching, two different equilibria produce this matching as the outcome of the game (one more restrictive than the other). In Subsection \ref{subseccion trabajadores tienen compromiso}, we consider a market where only workers hold commitment. Here, we also show that, given a stable matching, two different equilibria generate this matching as the outcome of the game (one more restrictive than the other). Note that when both sides of the market hold commitment, the dynamic game becomes a one-period static game where firms make offers, and workers either accept or reject, and the game then ends. In this setting, stable matchings are also supported by equilibria \citep[see][for more details]{alcalde1997hiring,alcalde2000simple}.

\subsection{The dynamic game without commitments}\label{subseccion sin compromiso} 

In this subsection, we study the case where no one has any commitment. When there is no commitment, the history leading to the current period does not change the continuation of the game and is therefore ignored by the agents in stationary equilibria. In other words, what happens in the current period does not affect future outcomes. Due to this independence, agents ignore the future and behave as in the static model.

The following result indicates that, in the absence of commitment, the static notion of stability captures the outcome of stationary equilibrium.

\begin{theorem}
    Given a dynamic matching game without commitment, the matching achieved in any stationary equilibrium is identical and stable in all periods. Conversely, for any stable matching, there exists a stationary equilibrium that achieves this matching in every period.
    \end{theorem}
 \begin{proof}
 Consider any stationary equilibrium $ \sigma$. The strategies of the firms $(\sigma_f)_{f \in F}$ in the equilibrium are stationary, i.e., firms always make the same offer ($ o^t_f = o^{t'}_f $ for any periods $ t, t' $). The action of the workers, that is, their response in any period, does not affect the offers they will receive in subsequent periods because they do not commit. This implies that the best decision each worker $ w $ can make is to accept the offer that provides the highest utility in each period. Thus, if the resulting matching of the equilibrium is $ \mu $, then it is the same in each period, i.e., $ \mu = \mu^t $ for each $ t $.

Suppose that equilibrium $ \sigma $ results in an unstable matching $ \mu $. Since $\mu$ is individually rational, let $ (f, w) $ be a blocking pair for $ \mu $. Since $ w $ blocks $ \mu $ together with $ f $, it follows that $ u_w(f) > u_w(\mu(w)) .$
Note that,  $ w $'s strategy is always to accept the best offer. This means she does not receive an offer from $ f $; otherwise, she would accept it.

Suppose that in some period $ t $, firm $ f $ deviates from the equilibrium $ \sigma $ and makes an offer to $ w $; that is, $ \hat{o}^t_f = \{w\} $. By the previous observation, the worker $ w $ will accept the offer and be matched with $ f $. Then, the firm obtains a higher profit by deviating and making an offer to its blocking pair, contradicting that the strategy $ \sigma $ is an equilibrium.

Conversely, let us choose a stable matching $ \mu $ and consider the following profile of strategies $ \sigma $ such that:
\begin{enumerate}[(i)]
    \item Each firm $ f $ makes an offer to $ \mu(f) $; that is, $ o^t_f = \{\mu(f)\} $ for each $t$.
    \item Each worker $ w $ accepts the offer that provides the highest utility.
\end{enumerate}

Note that the strategies are stationary since firms always make the same offer and workers accept the best among the received offers.

The workers' strategies are optimal because without commitment, accepting the offer that provides the highest utility will not affect future choices. The firms' strategies are also optimal because if a firm $ f $ makes an offer to a worker $ w \neq \mu(f) $ it is because $ u_f(w) > u_f(\mu(f)) $. Then, for worker $ w $ it must hold that $ u_w(\mu(w)) > u_w(f) $. Otherwise, $(f, w) $ will block $ \mu $ contradicting its stability. Furthermore, since $ w $ follows the equilibrium, the offer from $ f $ will be rejected. If the firm makes an offer to $ w $ such that $ u_f(\mu(f)) > u_f(w) $, the offer may be accepted and thus the firm's utility decreases implying that $f$ has no incentive to deviate. 
Therefore, $ \sigma $ is a stationary equilibrium. 
 \end{proof}

 The following Corollary is a consequence of the previous theorem and the Single Agent Theorem.\footnote{The Single Agent Theorem states that if an agent is single in a stable matching, it is single in all stable matchings \citep{mcvitie1970stable}.}
\begin{corollary}
     Workers and firms who do not have a match in a stationary equilibrium remain unmatched in all stationary equilibria.
\end{corollary}

\subsection{The dynamic game when firms hold commitment}\label{subseccion firmas tienen compromiso}

In this subsection, we examine the dynamic matching game where firms, that make the offers, are committed to their decisions. Workers are offered permanent positions, meaning they can resign but not be fired. This implies that if a worker wishes to improve their job situation, they must first resign from their current position and then wait for a new offer from a better firm. For the market under consideration, \cite{bonifacio2024counting} demonstrate that such an offer is always eventually made. They present a re-stabilization process and also provide a lower bound on the number of steps required for this process to reach a new stable matching. In our context, the number of steps in the re-stabilization process can be interpreted as the time a worker must wait to receive such an offer. This information, along with each worker's discount factor, will determine when a stable matching constitutes a stationary equilibrium.

Now, we examine the relationships between stationary equilibria and stable matchings depending on the discount factors of the workers.

Using a restrictive strategy regarding the firms' possible decisions, we demonstrate that any stable matching results from a stationary equilibrium. Then, the question arises whether stable matchings can result from another type of strategy that constitutes a subgame perfect Nash equilibrium. Given a stable matching and depending on the workers' patience, we show that such a matching can result from a stationary equilibrium.

\begin{theorem}\label{proposicion compromise firmas equilibria estacionario}
    Given a dynamic matching game where firms hold commitment, for every stable matching $ \mu $, there is a stationary equilibrium that achieves this matching.
\end{theorem} 
\begin{proof}
Let $ \mu $ be a stable matching, and consider the following strategy profile $ \sigma $ at each period $ t $:
\begin{enumerate}[(i)]
    \item Each firm $ f $ makes an offer to $ \mu(f) $; that is, $ o^t_f = \mu(f) $.
    \item Each worker $ w $ accepts the offer that provides the highest utility.
\end{enumerate} 

We will prove that the strategy $ \sigma $ is a stationary equilibrium.

Assume that $ f $, an active firm in the first period, deviates from $ \sigma $ by making an offer such that $ w \neq \mu(f)$ ($\hat{o}^1_f=w$), and  all other firms $ \hat{f} \neq f $ follow the strategy $ \sigma $ ($ o_{\hat{f} }= \mu(\hat{f}) $). We analyze the possible deviations for firm  $f$:
\begin{description}
    \item[\textbf{(i) $\boldsymbol{ \hat{o}^1_{f} = w $ and $ u_{f}(w) > u_{f}(\mu(f) )}$:}] Since $ \mu $ is stable, for the worker $ w $ it must hold that $ u_{w}(\mu(w)) > u_{w}(f) $. Otherwise, $ (f, w) $ would block $ \mu $ contradicting its stability. As workers play $ \sigma $, meaning they choose the offer that provides the highest utility, the offer from $ f $ will be rejected.
    
    \item[\textbf{ (ii) $\boldsymbol{ \hat{o}^1_{f} = w $ and $ u_{f}(\mu(f)) > u_{f}(w)} $:}] Then, $ w $ might accept the offer. Since firms hold commitment, they cannot layoff any worker, and therefore, the utility of firm $ f $ will be:
   $$
   \sum_{t=1}^\infty \delta_{f}^{t-1} u_{f}(w) < \sum_{t=1}^\infty \delta_{f}^{t-1} u_{f}(\mu(f))
   $$
   Hence, firm $ f$ will not benefit by making offer $ \hat{o}^1_{f}$.
   \item[\textbf{(iii) $\boldsymbol{\hat{o}^1_{f} =f }$, meaning that $\boldsymbol{f}$ deviates by making no offer:}] In the next period, the best outcome for $ f $ is to obtain $\mu(f)$, thus
   $$
   \sum_{t=2}^\infty \delta_{f}^{t-1} u_{f}(\mu(f)) < 
   \sum_{t=1}^\infty \delta_{f}^{t-1} u_{f}(\mu(f)).
   $$
   Therefore, the firm does not benefit by not making offers in period 1.
\end{description}

Now consider the strategies of workers. Firms do not deviate, meaning \( o_f = \mu(f) \) for all \( f \in F \). For each \( w \in W \), in every period \( t \), we have \( \mu(w) \in O^t_w \). Since \( \mu \) is stable, the offer that provides the highest utility to each \( w \) is \( \mu(w) \). Therefore, the workers' strategies are optimal, as without commitment, they accept the offer that maximizes their utility.

Thus, the strategy $ \sigma $ is a stationary equilibrium, whose outcome is $ \mu $.
    \end{proof} 
    
Note that the strategy $\sigma$ considered in Theorem \ref{proposicion compromise firmas equilibria estacionario} is restrictive with respect to the possible actions of the firms. Despite being rejected, a firm makes the same offer in the subsequent period. This raises the following question: Can stable matchings be achieved through a different type of strategy, allowing to firms behave differently? i.e. to make different offers? Fortunately, the answer is affirmative. However, before presenting the result that addresses this, we need to discuss the re-stabilization process introduced by \cite{bonifacio2024counting} and adapted to our context.

Assume that the market clears at a stable matching other than the worker-optimal one, leaving room for improvement from the workers’ perspective. We can interpret each period in a dynamic matching game where firms are committed to their offers as an iteration of the re-stabilization process presented in \cite{bonifacio2024counting}. In the first stage, each firm makes an offer to the worker assigned to them under the initial matching $\mu$, and each worker privately observes the offers she receives. Suppose that worker $w$ wishes to improve her labor situation and decides to adopt a strategy of resigning and waiting for a better offer. Consequently, worker $w$ rejects all offers, while the other workers accept the offer that provides them with the highest utility. The firm left unmatched in the previous stage, and then active, makes an offer to the next worker with a higher utility who has not previously rejected it and is willing to accept it, while the other firms repeat their previous offers. This process continues until worker $w$ receives a better offer.
Under this situation, the main result presented by \cite{bonifacio2024counting} establishes how many steps are required by the algorithm to re-stabilize the market.
The length of this process reflects the time the worker must wait to secure a new position. Understanding this timeline is crucial, as workers will be unemployed during this adjustment period, which affects their decision to resign or remain in their current position. In this way, we can determine whether the strategy of resigning and waiting for a better offer constitutes a stationary equilibrium.

Let $ w  \in W $ and $ \mu,\nu$ be two stable matchings where $\mu$ is the initial matching, and  $\nu$ is the re-stabilized matching when $w$ decides to improve her labor situation, provided by the re-stabilization process presented in \cite{bonifacio2024counting}. Denote by $ k(w ) $ the number of periods necessary for the firm $\nu(w )$ to make an offer to worker $ w  $ provided by the results in \cite{bonifacio2024counting}. Let $c^{w}$ defined as follows: 
$$ \label{cota factor de descuento}
c^{w } = \left( \frac{u_{w }(\mu(w ))}{u_{w }(\nu(w ))} \right)^{\frac{1}{k(w)}}.
$$ 

The following theorem guarantees that if the discount factor of each worker $w$ is bounded by $c^w$, then the stable matching $\mu$ results from a different stationary strategy, where firms have greater flexibility in their actions.

Let $\mu$ be a stable matching, and consider the following strategy profile $ \sigma $ at each period $ t $ in the dynamic game:
 \begin{enumerate}[(i)]
    \item For $t=1$, each firm $ f $ makes an offer a $\mu(f)$, i.e. $ o^1_f=\mu(f) $. For $t>1$, if $\mu^{t-1}(f)=\mu(f)$, then $o^t_f=\mu(f).$ Otherwise ($\mu^{t-1}(f)=\emptyset$), then $o^t_f=w$ where $w\in W$ is such that $u_w(f)>u_w(\mu^{t-1}(w))$ and $u_f(w)>u_f(w') $ for each $w'\in W\setminus\{w,\mu(f)\}.$  
    \item Each worker $ w $ accepts the offer that provides the highest utility.
\end{enumerate} 
Note that condition (i) for \(t > 1\) states that, if firm \(f\) is rejected by its partner under \(\mu\), it makes a new offer to a worker \(w\) who provides the highest utility among those workers who would not reject it.

\begin{theorem}
     Given a dynamic matching game where firms hold commitment, let $ \mu $ be a stable matching. If $ \delta_w \in (0, c^{w}) $ for each $ w \in W $, then $\sigma$ is a stationary equilibrium whose final matching is $ \mu $.
\end{theorem}
\begin{proof}
First, we prove that the strategy $ \sigma $ is a stationary equilibrium.
To do this, assume that $ f $, an active firm in the first period, deviates from $ \sigma $ by making an offer to some $\widehat{w}\in W$ such that $ \widehat{w} \neq \mu(f)$ ($\hat{o}^1_f=\widehat{w}$), and  all other firms $ \hat{f} \neq f $ follow the strategy $ \sigma $ ($ o_{\hat{f} }= \mu(\hat{f}) $). We analyze the possible deviations for firm  $f$:
\begin{description}
    \item[\textbf{(i) $\boldsymbol{\hat{o}^1_{f} = \widehat{w} $ and $ u_{f}(\widehat{w}) > u_{f}(\mu(f) )}$:}] Since $ \mu $ is stable, for the worker $ \widehat{w} $ it must hold that $ u_{\widehat{w}}(\mu(\widehat{w})) > u_{\widehat{w}}(f) $. Otherwise, $ (f, \widehat{w}) $ would block $ \mu $ contradicting its stability. As workers play $ \sigma $, meaning they choose the offer that provides the highest utility, the offer from $ f $ will be rejected.
    
    \item[\textbf{(ii) $\boldsymbol{\hat{o}^1_{f} = \widehat{w} $ and $ u_{f}(\mu(f)) > u_{f}(\widehat{w}) }$:}] Then, $ \widehat{w} $ might accept the offer. Since firms hold commitment, they cannot layoff any worker and, therefore, the utility of firm $ f $ will be:
   $$
   \sum_{t=1}^\infty \delta_{f}^{t-1} u_{f}(\widehat{w}) < \sum_{t=1}^\infty \delta_{f}^{t-1} u_{f}(\mu(f)).
   $$
   Hence, firm $ f$ will not benefit by making offer $ \hat{o}^1_{f}$.
   \item[\textbf{(iii) $\boldsymbol{ \hat{o}^1_{f} =f}$, meaning that $\boldsymbol{f }$ deviates by making no offer:}] In the next period, the best outcome for $ f $ is to obtain $\mu(f)$, thus
   $$
   \sum_{t=2}^\infty \delta_{f}^{t-1} u_{f}(\mu(f)) < 
   \sum_{t=1}^\infty \delta_{f}^{t-1} u_{f}(\mu(f)).
   $$
   Then, the firm does not benefit by not making offers in period 1.
\end{description}

Now, w.l.o.g. consider that firm $f$ inactive in the first period but rejected at the end of that period, i.e., $\mu^1(f)=\emptyset.$  Assume $f$ deviates from strategy $\sigma$, i.e., $\hat{o}^2_{f} \neq w$.
Thus, there are three possible deviations for firm  $f$ to analyze:

\begin{description}

 \item[\textbf{(i) $\boldsymbol{\hat{o}^2_{f} = \widehat{w} $ and $ u_{f}(\widehat{w}) > u_{f}(w )}$:}] for the worker \(\widehat{w}\), it must hold that \(u_{\widehat{w}}(\mu^1(\widehat{w})) > u_{\widehat{w}}(f)\). Otherwise, worker \(w\) is not the one who provides the highest utility among those who would accept an offer from \(f\), contradicting the definition of the strategy \(\sigma\). As the workers play \(\sigma\), meaning they choose the offer that provides the highest utility, the offer from \(f\) will be rejected.

  \item[\textbf{(ii) $\boldsymbol{\hat{o}^2_{f} = \widehat{w} $ and $ u_{f}(w) > u_{f}( \widehat{w} )}$:}] since workers have no commitment, then $ \widehat{w}$ accept the offer $\hat{o}^2_{f}$. Since firms hold commitment, they cannot layoff any worker, and therefore, the utility of firm $ f $ is:
   $$
   \sum_{t=2}^\infty \delta_{f}^{t-1} u_{f}(\widehat{w}) < \sum_{t=2}^\infty \delta_{f}^{t-1} u_{f}(w)
   $$
   Hence, firm $ f$ will not benefit by making offer $ \hat{o}^2_{f}$.
    \item[\textbf{(iii) $\boldsymbol{ \hat{o}^2_{f} =f}$, meaning that $\boldsymbol{f }$ deviates by making no offer:}] In the next period, the best outcome for $ f $ is to make offer  $\hat{o}^3_{f} =w$ and worker $w$ accept such an offer. Thus,
   $$
   \sum_{t=3}^\infty \delta_{f}^{t-1} u_{f}(w) < 
   \sum_{t=2}^\infty \delta_{f}^{t-1} u_{f}(w).
   $$
   Hence, the firm does not benefit by not making offers in period 2.
   
   \end{description}
Then, from cases (i)--(iii),  firm $f$ does not benefit by deviating in the second period and, therefore, follows the strategy $\sigma$ making offer $o^2_{f} = w$, where $w$ is the next worker on the list who will not reject it.

Now consider the strategies of the workers, i.e., each worker \( w \) accepts the offer that provides the highest utility. In this case, workers are the only ones that can deviate. For each \( w \in W \), in every period \( t > 1 \) where \(\mu^{t-1} = \mu\), we have \( \mu(w) \in O_w \). Since \( \mu \) is stable, the offer that provides the highest utility to each \( w \) is \( \mu(w) \). Assume that at stage \( t = 1 \), worker \( w \) decides to reject all offers (\( r^{1}_w = w \)). Following the re-stabilization process of \cite{bonifacio2024counting}, starting from matching $\mu$ when worker $w$ wants to improve her labor situation, there is a stable matching \( \nu \) such that \( u_{w}(\nu(w)) > u_{w}(\mu(w)) \). Let \( k(w) \) be the time that takes firm \( \nu(w) \) to make an offer to worker \( w \), i.e., \( \nu(w) \in O_w \). Worker $w$ will benefit from deviating by rejecting all offers in period 1 and waiting for the offer from \( \nu(w) \) if
\begin{equation}\label{ecu 1 teorema de la segunda estrategia}
    \sum_{t=1}^\infty \delta_{w}^{t-1} u_{w}(\mu(w)) < \sum_{t=1+k(w)}^\infty \delta_{w}^{t-1} u_{w}(\nu(w)).
\end{equation}
 By the resolution of the Geometric Series, we have\footnote{For the resolution of Geometric Series see \cite{stewart2012calculus}.}
   $$
   u_{w}(\mu(w))\frac{1}{1-\delta_{w}}  < u_{w}(\nu(w))\frac{\delta_{w}^{k(w)}}{1-\delta_{w}},
   $$
where $k(w)$ is the number of periods necessary for the firm $\nu(w )$ to make an offer to worker $ w  .$ Now, operating we obtain 
$$
   \frac{u_{w}(\mu(w))}{u_{w}(\nu(w))}  < \delta_{w}^{k(w)},
   $$
and, thus
$$
   \left(\frac{u_{w}(\mu(w))}{u_{w}(\nu(w))}\right) ^{\frac{1}{k(w)}}  < \delta_{w}.
   $$
That is, worker \( w \) will satisfy \eqref{ecu 1 teorema de la segunda estrategia} if her discount factor meets the condition \( \delta_{w} > \left(\frac{u_{w}(\mu(w))}{u_{w}(\nu(w))}\right)^{\frac{1}{k(w)}} \). Since by hypothesis we have \( \delta_{w} \in (0, c^{w}) \), where \( c^{w } = \left( \frac{u_{w}(\mu(w ))}{u_{w }(\nu(w ))} \right)^{\frac{1}{k(w)}} \), worker \( w \) does not benefit from deviating from $\sigma$.

Therefore, the strategy $ \sigma $ is a stationary equilibrium with outcome $ \mu $.   
    \end{proof}


\subsection{The dynamic game when workers hold commitment}\label{subseccion trabajadores tienen compromiso}

In this subsection, we examine the dynamic matching game where workers hold commitment, but firms do not. In this case, workers have job offers that are not permanent, i.e. although they cannot quit, they may be fired. Using a restrictive strategy concerning the firms' possible actions, we demonstrate that any stable matching is the outcome of a stationary equilibrium. Moreover, we prove that given a stable matching, and depending on firms' patience (i.e. considering a strategy of firms that is not restrictive as in the previous case), such a matching is the outcome of a stationary equilibrium.

\begin{theorem}\label{theorema con trabajadores comprometidos con estrategia restrictiva}
   Given a dynamic matching game where workers hold commitment, any stable matching $ \mu $ can be supported as a stationary equilibrium.
\end{theorem} 
\begin{proof}
Let $ \mu $ be a stable matching, and consider the following strategy profile $ \sigma $ at each period $ t $:
\begin{enumerate}[(i)]
    \item Each firm $ f $ makes an offer to $ \mu(f) $; that is, $ o^t_f = \mu(f) $.
    \item Each worker $ w $ accepts the offer that provides the highest utility.
\end{enumerate} 

We will prove that the strategy $ \sigma $ is a stationary equilibrium.

Assume that until period $\widetilde{t}-1>2$, all agents play the strategy profile $\sigma$. Also, assume that there is a firm that deviates at stage $\widetilde{t}$, i.e., $\hat{o}^{\widetilde{t}}_f \neq \mu(f)$. Now, we have three cases to consider:

\begin{description}

\item [\textbf{(i) $\boldsymbol{\hat{o}^{\widetilde{t}}_{f} = w}$, $ u_{f}(w) > u_{f}(\mu(f)) $, and $\mu^{\widetilde{t}-1}(w)\neq w$:}] Since workers hold commitment, and the other firms are playing $\sigma$, $w$ rejects the offer $\hat{o}^{\widetilde{t}}_{f}$.

\item [\textbf{(ii) $\boldsymbol{\hat{o}^{\widetilde{t}}_{f} = w}$, $ u_{f}(w) > u_{f}(\mu(f)) $, and $\mu^{\widetilde{t}-1}(w)= w$:}] Since $\mu$ is stable and agents have been playing $\sigma$ until stage $t-1$, we have that $0 = u_w(\mu(w)) > u_w(w)$. Then, $w$ rejects the offer $\hat{o}^{\widetilde{t}}_{f}$.

\item [\textbf{(iii) $\boldsymbol{ \hat{o}^{\widetilde{t}}_{f} = f}$, meaning that $\boldsymbol{f}$ deviates by making no offer:}] In period $\widetilde{t}+1$, since the rest of the agents are playing $\sigma$, the best outcome for $f$ is to obtain $\mu(f)$, thus
   $$
   \sum_{t=1}^{\widetilde{t}-1} \delta_{f}^{t-1} u_{f}(\mu(f)) +\sum_{t=\widetilde{t}+1}^\infty \delta_{f}^{t-1} u_{f}(\mu(f))< 
   \sum_{t=1}^\infty \delta_{f}^{t-1} u_{f}(\mu(f)).
   $$
   Therefore, the firm does not benefit by not making offers in period $\widetilde{t}$.
\end{description}
Then, from cases (i)--(iii),  firm $f$ does not benefit by deviating in period $\widetilde{t}$ and, therefore, follows the strategy $\sigma$ making offer $o^{\widetilde{t}}_{f} = \mu(f)$.

Now, consider the strategies of the workers. Firms do not deviate, meaning \( o_f = \mu(f) \) for all \( f \in F \). For each \( w \in W \), in every period \( t \), we have \( \mu(w) \in O^t_w \). Since \( \mu \) is stable, the offer that provides the highest utility to each \( w \) is \( \mu(w) \). Therefore, the workers' strategies are optimal, as without commitment, they accept the offer that maximizes their utility.

Therefore, the strategy $ \sigma $ is a stationary equilibrium, whose outcome is $ \mu $.
    \end{proof} 

Note that, as in the previous section, the strategy $\sigma$ considered in Theorem \ref{theorema con trabajadores comprometidos con estrategia restrictiva} is restrictive with respect to the
possible actions of the firms. Despite being rejected, a firm makes the same offer in
the subsequent period. Again, this fact opens the same question as before: Can stable matchings be
achieved through a different type of strategy, allowing to firms behave differently?
i.e. to make different offers? Fortunately, again the answer is affirmative.

Let $ f \in F $. Let $c^{f}$ defined as follows: 
$$ 
c^{f } = \left( \frac{u_{f }(\mu(f ))}{u_{f }(\mu_{F}(f ))} \right).
$$ 

The following theorem guarantees that if the discount factor of each firm $f$ is bounded by $c^f$, then the stable matching $\mu$ results from a different stationary strategy, where firms have more flexibility in their actions.

Let $\mu$ be a stable matching, and consider the following strategy profile $ \sigma $ at each period $ t $ in the dynamic game:
 \begin{enumerate}[(i)]
    \item For $t=1$, each firm $ f $ makes an offer a $\mu(f)$, i.e. $ o^1_f=\mu(f) $. For $t>1$, if each firm $f$ observe that $\mu^{t-1}=\mu$, then $o^t_f=\mu(f).$ Otherwise (if each firm $f$ observe $\mu^{t-1}\neq \mu$), then $o^t_f=\mu_F(f).$ 
    \item Each worker $ w $ accepts the offer that provides the highest utility.
\end{enumerate} 
Note that condition (i) for \(t > 1\) states that, if firm \(f\) observes that in the previous period, the resulting matching is not \(\mu\), the offer it makes is $\mu_F(f)$, i.e. its partner under the firm-optimal stable matching.

\begin{theorem}
  Given a dynamic matching game where workers hold commitment, let $ \mu $ be a stable matching. If $ \delta_f \in (0, c^{f}) $ for each $ f \in F $, then $\sigma$ is a stationary equilibrium supporting the stable matching $ \mu $.
\end{theorem}
\begin{proof}
Let $\mu$ be a stable matching. We will prove that the strategy $ \sigma $ is a stationary equilibrium that supports $\mu$ when $ \delta_f \in (0, c^{f}) $.

Assume that until period $\widetilde{t}-1>2$, all agents play the strategy profile $\sigma$. Also, assume that there is a firm that deviates at stage $\widetilde{t}$, i.e., $\hat{o}^{\widetilde{t}}_f \neq \mu(f)$. Now, we have three cases to consider:

\begin{description}

\item [\textbf{(i) $\boldsymbol{\hat{o}^{\widetilde{t}}_{f} = w$, $ u_{f}(w) > u_{f}(\mu(f)) $, and $\mu^{\widetilde{t}-1}(w)\neq w}$:}] Since workers hold commitment, and given that other firms in this period are playing $\sigma$ and the outcome matching from the previous period is $\mu$, each firm $\hat{f} \neq f$ makes an offer to $\mu(\hat{f})$, i.e., $o^{\widetilde{t}}_{\hat{f}} = \mu(\hat{f})$. Consequently, $w$ rejects the offer $\hat{o}^{\widetilde{t}}_{f}$.

\item [\textbf{(ii) $\boldsymbol{\hat{o}^{\widetilde{t}}_{f} = w$, $ u_{f}(w) > u_{f}(\mu(f)) $, and $\mu^{\widetilde{t}-1}(w)= w}$:}] Since $\mu$ is stable and agents have been playing $\sigma$ until stage $\widetilde{t}-1$, we have that $0 = u_w(\mu(w)) > u_w(f)$. Then, $w$ rejects the offer $\hat{o}^{\widetilde{t}}_{f}$.

\item [\textbf{(iii) $\boldsymbol{ \hat{o}^{\widetilde{t}}_{f} = f}$, meaning that $\boldsymbol{f}$ deviates by making no offer:}] In period $\widetilde{t}+1$, since the rest of the agents are playing $\sigma$ and the resulting matching in period $\widetilde{t}$ differs from $\mu$, each firm $\hat{f} \neq f$ makes an offer to $\mu_F(\hat{f})$, i.e., $o^{\widetilde{t}+1}_{\hat{f}} = \mu_F(\hat{f})$. Firm $f$ would benefit from deviating by not making an offer, and the best strategy it can follow is $\mu_F(f)$ if
\begin{equation}\label{ecu 1 teorema de la segunda estrategia para f}
    \sum_{t=1}^\infty
    \delta_{f}^{t-1} u_{f}(\mu(f)) <  \sum_{t=1}^{\widetilde{t}-1} \delta_{f}^{t-1} u_{f}(\mu(f)) +\sum_{t={\widetilde{t}+1}}^\infty \delta_{f}^{t-1} u_{f}(\mu_F(f)).
\end{equation}
 By the resolution of the Geometric Series, we have\footnote{For the resolution of Geometric Series see \cite{stewart2012calculus}.}
   $$
   u_{f}(\mu(f))\frac{1}{1-\delta_{f}}  <  u_{f}(\mu(f)) \left( \frac{1-\delta_{f}^{\widetilde{t}-1}}{1-\delta_{f}} \right) +  u_{f}(\mu_F(f))\frac{\delta_{f}^{\widetilde{t}}}{1-\delta_{f}}.
   $$    
  By operating the previous inequality, we obtain 
    $$
   u_{f}(\mu(f)) <  u_{f}(\mu(f)) \left( 1-\delta_{f}^{\widetilde{t}-1} \right) +  u_{f}(\mu_F(f))\delta_{f}^{\widetilde{t}},
   $$
  and thus  
  $$
   u_{f}(\mu(f)) - u_{f}(\mu(f)) + u_{f}(\mu(f))\delta_{f}^{\widetilde{t}-1}  < u_{f}(\mu_F(f))\delta_{f}^{\widetilde{t}}.
   $$
Then, operating again, we obtain 
$$
   \frac{u_{f}(\mu(f))}{u_{f}(\mu_F(f))}  < \delta_{f}.
   $$
That is, firm \( f \) will satisfy \eqref{ecu 1 teorema de la segunda estrategia para f} if their discount factor meets the condition \( \delta_{f} >  \frac{u_{f}(\mu(f))}{u_{f}(\mu_F(f))}  \). Since by hypothesis we have \( \delta_{f} \in (0, c^{f}) \), where \( c^{f } = \left( \frac{u_{f}(\mu(f))}{u_{f }(\mu_F(f ))} \right) \), firm \( f \) does not benefit from deviating from $\sigma$.
\end{description}

Then, from cases (i)--(iii),  firm $f$ does not benefit by deviating in period $\widetilde{t}$ and, therefore, follows the strategy $\sigma$ making offer $o^{\widetilde{t}}_{f} = \mu(f)$.

Now, consider the strategies of the workers. Firms do not deviate, meaning \( o^t_f = \mu(f) \) for each \( f \in F \) and each period $t$. Hence, \( \mu(w) \in O^t_w \) for each \( w \in W \). Since \( \mu \) is stable, the offer that provides the highest utility to each \( w \) is \( \mu(w) \). Therefore, the workers' strategies are optimal, as in the case that workers have no commitment, they accept the offer that maximizes their utility.

Therefore, the strategy $ \sigma $ is a stationary equilibrium, whose outcome is $ \mu $.
 \end{proof} 

\section{Concluding remarks}\label{section concludings}

This paper extends the static matching model by \cite{gale1962college} to a dynamic setting where firms and workers engage repeatedly, examining how different commitment structures and levels of patience among agents impact long-term stability in decentralized markets. We model a non-cooperative dynamic game where firms periodically offer positions, and workers choose to accept or reject these offers, without forming beliefs about others' actions. The study considers three commitment scenarios---both sides, firms-only, and workers-only---and introduces stationary equilibrium as the solution concept. We find that stable matchings can be supported as stationary equilibria under specific commitment settings, providing insights into how varying levels of patience and commitment affect equilibrium and stability outcomes over time.

If we consider a potential extension to a many-to-one model, where firms can hire multiple workers, the game becomes asymmetric regarding which agents make the offers. Due to this asymmetry, our results do not extend directly, even assuming responsive utilities for firms. For example, in cases where firms hold commitment, it is not possible to determine the exact number of stages required for re-stabilization; only a lower bound can be provided \citep[see][for more details]{bonifacio2024counting}. Therefore, considering these market characteristics, an indirect extension may be feasible for many-to-one markets where firms exhibit responsive or even substitutable utility functions.

\end{document}